\documentclass[a4paper, 12pt]{article}
\usepackage[british]{babel}
\usepackage{amsmath,amsfonts}
\usepackage[square,numbers]{natbib}
\usepackage[T1]{fontenc}
\usepackage[dvips]{graphicx}
\usepackage{xcolor}
\usepackage{amsmath, amsfonts,amsthm, mathrsfs, amssymb}
\usepackage{mathtools}

\usepackage{enumerate}
\usepackage[latin1]{inputenc}
\usepackage{geometry}

\numberwithin{equation}{section}
\newtheorem{theorem}{Theorem}[section]
\newtheorem{lemma}{Lemma}[section]
\newtheorem{corollary}{Corollary}[section]
\newtheorem{proposition}{Proposition}[section]
\newtheorem{definition}{Definition}[section]
\newtheorem{remark}{Remark}[section]

\def\ndiff#1#2#3{{\frac{{\rm d}^{#1} #2}{{\rm d} #3^{#1}}}}
\def\diff#1#2{{\frac{{\rm d} #1}{{\rm d} #2}}}
\def\nparder#1#2#3{{\frac{\partial^{#1} #2}{\partial #3^{#1}}}}
\def\parder#1#2{{\frac{\partial #1}{\partial #2}}}

\def\Dscr{\mathcal{D}}

\def\cI{{\mathcal{I}}}
\def\cP{{\mathcal{P}}}
\def\cE{{\mathcal{E}}}
\def\cF{{\mathcal{F}}}
\def\cS{{\mathcal{S}}}
\def\cA{{\mathcal{A}}}
\def\cL{{\mathcal{L}}}
\def\cK{{\mathcal{K}}}
\def\cO{{\mathcal{B}}}
\def\cC{{\mathcal{C}}}
\def\cU{{\mathcal{U}}}
\def\cD{{\mathcal{D}}}
\def\cV{{\mathcal{V}}}
\def\cM{{\mathcal{M}}}
\def\cW{{\mathcal{W}}}
\def\bx{{\bf x}}
\def\bp{{\bf p}}
\def\bL{{\bf L}}

\def\Ve{V_{{\rm eff}}}
\def\Veff{V_{{\rm eff}}}
\def\te{{\tt E}}

\def\reali{\mathbb{R}}
\def\toro{\mathbb{T}}
\def\rimuovi#1{}
\def\hze{H}
\def\heff{H_{{\rm eff}}}
\def\cA{{\cal A}}
\def\cS{{\cal S}}
\def\R{\mathbb{R}}

\begin{document}
\title{Exponential stability in the perturbed 
  central force problem. }

\author{D. Bambusi\thanks{
    \noindent Dipartimento di Matematica ``Federigo Enriques'', Universit\`a degli Studi di Milano,
    Via Saldini 50, 20133 Milano.\newline%
    {\it Emails}: \vtop{\hbox{\tt dario.bambusi@unimi.it}\hbox{\tt alessandra.fuse@unimi.it}\hbox{\tt marco.sansottera@unimi.it}}}, A. Fus\`e\footnotemark[1], M. Sansottera\footnotemark[1]}

\date{\today}

\maketitle

\begin{abstract}
We consider the spatial central force problem with a real analytic
potential. We prove that for all analytic potentials, but the Keplerian and the
Harmonic ones, the Hamiltonian fulfills a nondegeneracy property
needed for the applicability of Nekhoroshev's theorem. We deduce
stability of the actions over exponentially long times when the system
is subject to arbitrary analytic perturbation. The case where the
central system is put in interaction with a slow system is also
studied and stability over exponentially long time is proved.
\end{abstract}

\section{Introduction and statement of the main result}\label{sec:intro}

It is well known that, if the Hamiltonian of an integrable system
fulfills some nondegeneracy properties, then Nekhoroshev's theorem
applies and yields strong stability properties of small perturbations
of the system \cite{Nekhoroshev-1977,Nekhoroshev-1979}. Precisely one
gets that the actions are approximately integrals of motions for times
exponentially long with the inverse of the size of the perturbation. A
sufficient nondegeneracy condition ensuring applicability of
Nekhoroshev's Theorem is quasiconvexity of the Hamiltonian in
action-angle coordinates. However, it is nontrivial to prove that in a
specific system such a condition is verified, and actually this is
known only for a few systems.

In this paper, continuing the investigation of \cite{RCD}, we study
the applicability of Nekhoroshev's theorem to perturbations of the
spatial central force problem. The spatial central motion is
a superintegrable system that does not admit global
action-angle coordinates and moreover its Hamiltonian depends on two
actions only. As a consequence it is never quasiconvex in the standard
sense. However, modern versions of Nekhoroshev's theorem
\cite{Fas95,Blaom} apply provided the Hamiltonian is quasiconvex as a
function of the two actions on which it actually depends. Here we
prove that for {any} analytic central potential (fulfilling
assumptions (H1-H3) below) except the Keplerian
and the Harmonic ones, this is the case. Therefore, one gets stability
over exponentially long times of these two actions.

To come to a formal statement consider the Hamiltonian of the spatial
central force problem in Cartesian coordinates
  \begin{equation}
    \label{carte}
    \hze (\bx,\bp)=\frac{\left|\bp\right|^2}{2}+V(|\bx|)\ ,
  \end{equation}
  where $\bx\equiv(x,y,z)$, $\bp\equiv(p_x,p_y,p_z)$,
  $\left|\bx\right|= \sqrt{x^2+y^2+z^2}$ and the potential
  $V\colon\reali_{>0}\to\reali$ is a {\it real analytic function}
  satisfying
  \begin{itemize}
  \item[(H1)] $\displaystyle{-\ell^*\coloneqq\lim_{r\to0^+}
    \frac{r^2V(r)}{2}>-\infty}${\rm\thinspace;}
  \item[(H2)] $\displaystyle{\exists r>0 \,:\, r^3V'(r)>\max\left\{\,
    0,\,\ell^*\,\right\}}${\rm\thinspace;}
  \item[(H3)] $\forall \ell>\max\left\{\, 0,\, \ell^*\,\right\}$,
    $r^3V'(r)=\ell$ has a finite number of solutions (in $r$)\/.
  \end{itemize}

\begin{remark}
\label{rem.H}
Assumption (H1) ensures that there exist values of the
  angular momentum s.t. the effective potential tends to plus
infinity as $r\to0^+$, while (H2) ensures that the effective potential
has at least one minimum. As a consequence the domain of the
action-angle variables is not empty.

These assumptions could be removed and substituted by the assumption
that the domain of the action-angle variables is not empty. Here we
decided to avoid such a vague assumption, substituting it with
some other one, easy to verify and guaranteeing such a condition.
\end{remark}
\begin{remark}
\label{H3}
 Assumption (H3) is a technical assumption needed in order to simplify
 the proof of Lemma~\ref{lemma4.1} giving the structure of the domain of action-angle variables. We expect that this assumption can be removed.
\end{remark}

We denote by $\bL:=\bx\times\dot \bx$ the total
angular momentum and by $L=|\bL|$ its modulus.
Then we consider the perturbed Hamiltonian system
  $$
  H_\epsilon=\hze +\epsilon P\ ,
  $$ were $P$ is analytic for $\bx\not=0$.  For its
  dynamics the following theorem holds.

\begin{theorem}
  \label{Nek.1}
 Assume that the potential $V$ fulfills (H1)--(H3) and that it is neither
 Harmonic nor Keplerian, then there exists a set $\cK\subset\R^6$,
 which is the union of finitely many analytic hypersurfaces such that:
 given a compact set $\cC\subset \R^6\setminus\cK$, invariant for the
 dynamics of $\hze $, then there exist positive constants
 $\epsilon^*$, $C_1$, $C_2$, $C_3$ and $C_4$ such that, for
 $|\epsilon|<\epsilon^*$ and initial data in $\cC$, along the dynamics
 of $H_\epsilon$, one has
   \begin{equation}
    \label{est.nek}
    \left|L(t)-L(0)\right|\leq C_1\epsilon^{1/4}\ ,\quad
    \left|\hze (t)-\hze (0)\right|\leq C_2\epsilon^{1/4} \ ,
  \end{equation}
  for 
  \begin{equation}
    \label{time}
    \left|t\right|\leq C_3\exp(C_4\epsilon^{-1/4})\ .
  \end{equation}
\end{theorem}

\begin{remark}
\label{kappa}
The set $\cK$ is the set where the nondegeneracy assumptions are not
fulfilled.  In the statement of the theorem we consider a proper
subset $\cC$ in order to obtain values of the nondegeneracy constants
independent of the initial datum. A slightly different statement
avoiding $\cC$ is the following one.
\end{remark}

\begin{theorem}
\label{senzaC}
Assume that the potential is neither Harmonic nor Keplerian, then,
given an initial datum in $\R^6\setminus\cK$, there
exist constants $\epsilon^*$, $C_1$, $C_2$ , $C_3$ and $C_4$ such
that, for $|\epsilon|<\epsilon^*$, along the
dynamics of $H_\epsilon$, the inequalities \eqref{est.nek} hold for
the time \eqref{time}.
\end{theorem}

An immediate consequence of the above theorem is that, for an
exponentially long time, the particle's orbits are confined between
two spherical shells centered at the origin.  In addition, in the
spirit of~\cite{vincoli1,vincoli2,vincolilombardi}, in
Section~\ref{stat} we will also give a stability result for the
case where the central system is put in interaction with a ``slow''
system as in models of molecular dynamics \cite{BG93,teufel} (see
Theorem~\ref{nek.2}). We will also consider the application to the
dynamics of a soliton in NLS~\cite{BM16}.

Our result is strongly connected with Bertrand's theorem, according to
which, if all the bounded trajectories of a particle moving under the
action of a central force are periodic, then the force is either
Keplerian or Harmonic~\cite{Bertrand,FK04}. Actually some steps of our
procedure are reminiscent of the classical proof of Bertrand's
theorem, yet our analysis also yields a new proof of a version of
Bertrand's theorem stronger than the classical one (see
Appendix~\ref{Bert}).

We compare now our result to the result of \cite{RCD}. The present
result is stronger than that of \cite{RCD} in two aspects.  
First, in \cite{RCD} the authors proved that non-quasiconvex potentials
fulfill a suitable differential equation and that among the
homogeneous potentials only the Harmonic and the Keplerian ones
fulfill such an equation. Here, through a deeper analysis, we show
that {\it all} the analytic potential fulfilling assumption
(H3)\footnote{One should add ``and also (H1) (H2)'', but, as explained
  in Remark \ref{rem.H} these are not essential.} give rise to quasi
convex Hamiltonians except the the Harmonic and the Keplerian ones.

Secondly, the result of \cite{RCD} was only local and the existence of
a nondegenerate minimum of the effective potential was explicitly
assumed and only initial data belonging to a neighbourhood of such a
minimum were considered. On the contrary, we give here a result
allowing to deal with {\it any} initial datum giving rise to bounded
orbits. In particular, since according to our assumptions the
effective potential could have local maxima (as in Figure 1), the
domain of existence of action-angle variables can have also a
connected component not containing a minimum of the potential. The
study of this region requires a completely new analysis that is the
object of Subsect.~\ref{sec:max}. Making reference to Figure 1, the
result of \cite{RCD} only allows to deal with the regions strictly
below $V_2$, while here we also deal with the region between $V_2$ and
$V_{\infty}$.

\vskip10pt

We now outline the strategy adopted for proving our main result. First
we show that the result for the spatial central force problem follows
from standard quasiconvexity of the Hamiltonian of the {\it planar}
central motion.  We stress that even if the proof is based on the
study of the planar central motion, our main dynamical result, i.e.,
Theorem~\ref{Nek.1}, pertains the spatial case. Of course one can
deduce a similar result also for the planar case, but in this case a
much stronger result would follow from the application of KAM theory.
Indeed, since in systems with two degrees of freedom the KAM tori
separate the phase space, their existence implies stability of the
actions for all times.  Furthermore, the proof of Bertrand's theorem given
in~\cite{FK04} also yields applicability of KAM theorem to all the
analytic cases which are neither Harmonic nor Keplerian.

The planar central force problem has two degrees of freedom and the key
remark is that for systems with two degrees of freedom quasiconvexity
is equivalent to the non vanishing of the Arnold determinant.
Furthermore, in an analytic context the Arnold determinant is an
analytic function of the actions, therefore only two possibilities occur:
either it is a trivial analytic function, or it is always different
from zero except on an analytic hypersurface.  In Theorem~\ref{bert}
we will prove that for the Hamiltonian of the planar central motion
such a determinant is a nontrivial function, except in the
Keplerian and Harmonic cases.

The study of the Arnold determinant is done via its asymptotics at
circular orbits.  Precisely, we introduce polar coordinates and recall
that the two actions $(I_1,I_2)$ of the planar central motion are the
modulus of the angular momentum $p_{\theta}\equiv I_2$ and the action
of the effective system describing the motion of the radial variable
\begin{equation}
  \label{H}
\heff (r,p_r,p_{\theta})\coloneqq\frac{p_r^2}{2}+\Ve(r,p^2_\theta)\ ,
\end{equation}
where 
\begin{equation*}
\Ve(r,p_\theta^2)\coloneqq V(r)+\frac{p_\theta^2}{2r^2}\ ,
\end{equation*}
and $p_\theta$ plays the role of a parameter.

The critical points of the effective potential correspond to circular
orbits of the planar system and we show that, except for at most a
finite number of values of $p_\theta$ that we eliminate, such critical
points are nondegenerate.  This property is crucial in the subsequent
analysis.

The regions where the action $I_1$ is defined are open connected
regions of the space $(r,p_r)$, that we denote by $\cE_j$ ($j$ being 
an index ranging over a finite set), which are constructed as unions
of connected components of compact level surfaces of $\hze $.  We
distinguish two situations: (1)~the closure of $\cE_j$ contains a
minimum of the effective potential; (2)~the closure of $\cE_j$ does
not contain a minimum of the effective potential, thus it must contain
a maximum.

First we analyze the regions~(1) following the approach adopted in~\cite{RCD}.
The idea is to exploit the fact that, for one-dimensional systems, the
Birkhoff normal form is convergent close to a nondegenerate minimum of the
potential.  Thus, at least in principle, one can effectively introduce
action-angle variables via Birkhoff normal form and compute the expansion of the
Arnold determinant at a minimum. In~\cite{RCD}, using the expansion of the
Hamiltonian at order two in the actions, it was shown that the zero order
term of the Arnold determinant vanishes if and only if the potential fulfills a
suitable forth order differential equation.  In the present paper, pushing
further the expansions with the aid of the Mathematica$^{TM}$ algebraic
manipulator, we expand the Hamiltonian up to order four and we show that the
first order term of the Arnold determinant vanishes if the potential fulfills a
suitable sixth order differential equation. Finally, still using symbolic
manipulation, we show than only the Keplerian and Harmonic potentials satisfy
these differential equations.  Thus, these are the only two cases in which
the expansion of the Arnold determinant at the minimum vanishes.

Second we study the regions (2).  Here we compute the asymptotic
expansion of the Arnold determinant at the maximum and prove that it
is divergent. Thus, if the effective potential has a maximum, then the
Hamiltonian is almost everywhere quasiconvex in the region above it.
The formula for the behavior of the actions at the maximum has been
already obtained by Neishtadt in \cite{neishtadt1987change}.  Here we
give an independent proof based on a normal form result for one
dimensional Hamiltonians at saddle points due to Giorgilli
\cite{Gio01}.  We remark that similar formul{\ae} have also been
re-derived by Biasco and Chierchia in \cite{2017arXiv170206480B} for
the study of the measure of invariant tori in nearly-integrable
systems.

\noindent{\it Acknowledgments.} We thank Francesco Fass\`o for
suggesting the connection between the result of~\cite{RCD} and the
Bertrand's theorem; Francesco De Vecchi for a hint on the method to
find the common solutions of two differential equations; Anatoly
Neishtadt for suggesting the formula giving the expansion of the
action variable at a maximum, namely~\eqref{I1finale}; Luca Biasco for
a discussion on the asymptotic expansion of the actions at maxima of
the potential and Michela Procesi for suggesting to use our strategy
for a new proof of Bertrand's theorem.

\section{Generalized action-angle variables}\label{statement}
First of all we recall that in general a superintegrable system does
not admit global action-angle variables. This comes from the fact
that, in general, the subset of the phase space corresponding to a fixed
value of the actions has a nontrivial topology.  

A general theory of superintegrable systems has been developed in
\cite{Fas95, Fas05, Ka012, MF78}. Applying such a theory and using the
specific form of the Hamiltonian we are going to prove a theorem
(Lemma~\ref{az.ang.3}) giving the structure of the phase space and
describing generalized action-angle variables for the spatial central force
problem. This is a refinement of Theorems~1 and~2 in~\cite{RCD}.

\vskip 5pt

Consider again the Hamiltonian~\eqref{carte}.  In order to define the set where
the modulus of the angular momentum varies, we introduce
$$
\begin{aligned}
L_m^2&\coloneqq\min\Bigl\{\reali_{>0}\cap\left[\text{{\rm Range}}(r^3V'(r))\right]\cap[\ell^*,+\infty]
\Bigr\}\ , \\
L_M&\coloneqq
\begin{cases}
  \sup \sqrt{r^3V'(r)}\ ,& \text{if } \sup r^3V'(r)<+\infty\ ;\\
  \hat{L}_M\ ,& \text{if } \sup r^3V'(r)=+\infty\ ;\\
\end{cases}
\end{aligned}
$$
where $\hat{L}_M$ is an arbitrary large positive number.

The modulus of the angular momentum will be assumed to vary in
\begin{equation*}
\cI\coloneqq(L_m,L_M)\ .
\end{equation*}

\begin{remark}
\label{initial}
One can think of studying the orbit of the system corresponding to
a given initial datum, in the spirit of Theorem
  \ref{senzaC}. Correspondingly, having fixed an initial datum giving
rise to bounded orbits, one can compute the value of $L_0$ of the
angular momentum and can take $L_M>L_0$. Of course one can expect all
the constants to deteriorate as $L_M$ increases.
\end{remark}

The Hamiltonian in action-angle variables has the same form as in the
planar case, so to come to a precise statement, which will be useful
for the proof of Theorem \ref{Nek.1}, we start to study the planar
case.

Consider the planar case and introduce polar coordinates. Then the
action variables are $I_2:=p_\theta$ and the action $I_1$ of the
effective radial system with Hamiltonian $\heff$. The domains of
definition of $I_1$ are determined by the sublevels of $\heff$, which
in turn are determined by $\Veff$. First we want to have stability of
the sublevels for small changes of $p_\theta$. To this end we will
prove that for all values of $p_\theta$ except at most a finite
number, $\Veff$ is a Morse function. Thus one gets that $\cI$ can be
decomposed into the union of a finite number of open segments $\cI_j$ plus a
finite number of points, with the property that the shape of $\Veff$
does not change for $p_\theta\in\cI_j$ and furthermore all the
critical points of $\Veff$ are nondegenerate. We will show that one
can also select the intervals $\cI_j$ in such a way that in such
intervals all the critical levels of $\Veff$ are distinct. 

\begin{figure}
  \centering
  \includegraphics[width=0.8\textwidth]{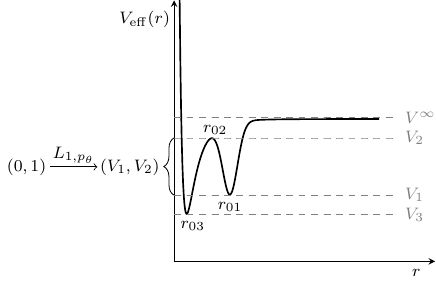}
  \caption{A double-well interaction potential (Lennard-Jones-Gauss).}
  \label{fig1}
\end{figure}

Then one has to introduce the action $I_1$. In order to describe the
procedure we make reference to a particular shape depicted in
Fig.~\ref{fig1}. First we remark that, as $p_\theta$ varies in
$\cI_j$, the critical points of $\Veff$ can move, but they cannot cross and their
levels cannot cross neither. So, fix an arbitrary value of
$p_\theta\in\cI_j$. Consider the sublevel $\heff<V_2$, then the first
domain where one can introduce action-angle coordinates is its
connected component containing $r_{01}$, to which one has to eliminate
the point $(p_r,r)=(0,r_{01})$. So, such a domain is isomorphic to an
interval $E\in (V_1,V_2)$ times the level curve $\heff=E$. In order to
make clear the topological structure of the domains, in the statement
of the forthcoming theorem we also introduce an affine map
$L_{1,p_{\theta}}$ which transforms the interval $(0,1)$ into
$(V_1,V_2)$. Of course the map (as well as the target interval) depend
on the value of $p_\theta$ as well as on the domain identified by the
critical point $r_{0i}$ ($r_{01}$ in our case).  Then one repeats the
construction in all the other domains getting the complete picture.

To give a precise $2-d$ statement, define, for $p_{\theta}\in \cI$ and arbitrary $E\in\R$, 
\begin{equation*}
\cL^{(2)}(p_{\theta},E)\coloneqq\bigl\{\, (r,p_r)\, :\, \heff (r,p_r,p_{\theta})=E\,\bigr\}\ .
\end{equation*}
As $E\in\R$ and $p_{\theta}\in\cI$ vary, the sets
$\cL^{(2)}(p_{\theta},E)$ can be empty or can have one or more
connected components. For fixed $E$, we enumerate the nonempty compact
connected components by $\cL_i^{(2)}(p_{\theta},E)$.

Defining
\begin{equation*}
\cP_A^{(2)}\coloneqq\left\{(r,p_r,\theta,p_\theta)\ :\ \theta\in\toro\ ,\ p_\theta\in\cI\ ,\ (r,p_r)\in
\bigcup_{E\in\R}\bigcup_i\cL^{(2)}_{i}(p_\theta,E)\right\} \ .
\end{equation*}
we have the following lemma.

\begin{lemma}
\label{az.ang}
There exists a finite number of open intervals $\cI_j$ such that, for
any $p_{\theta}\in\cI_j$ there exist a finite number of affine maps,
$L_{i,p_\theta}\colon(0,1)\to\R$ with the following property:
define
\begin{equation*}
  \cO_{i,j}^{(2)}\coloneqq  \left\{(r,p_r,\theta,p_\theta)\, :\, \theta\in\toro\ ,\ p_\theta\in\cI_j\ ,\ (r,p_r)\in
\cL^{(2)}_{i}(p_\theta,L_{i,p_\theta}\te)\ ,\  \te\in (0,1)\right\}\ ,
\end{equation*}
and $\cS^{(2)}\coloneqq\cP_A^{(2)}\setminus \bigcup
_{i,j}\cO_{i,j}^{(2)}$; then
\begin{itemize}
\item [(i)] 
$\cS^{(2)}$ is the union of a finite
  number of analytic hypersurfaces;
\item[(ii)]
  on each of the
domains $\cO_{i,j}^{(2)}$ there exists an analytic diffeomorphism
\begin{align*}
\Phi_{i,j}^{(2)}:\cO_{i,j}^{(2)}&\to\cA_{i,j}\times \toro^2\ ,\qquad \cA_{i,j}\subset\R^2
\\
(r,p_r,\theta,p_{\theta})&\mapsto (I_1,I_2,\alpha_1,\alpha_2)
\end{align*}
which introduces action-angle variables;
\item[(iii)]
for fixed
$p_\theta\in\cI_j$, the infimum of $\heff$ over
$$
\{(r,p_r)\ :\ (r,p_r,\theta,p_{\theta})\in \cO_{i,j}^{(2)}\}
$$
is either a nondegenerate minimum or a
nondegenerate maximum of $\Ve(.;p_\theta^2)$. The function $\heff $
has no other critical points in $\cO_{i,j}^{(2)}$. 
\end{itemize}
\end{lemma}
This Lemma will be proved in Section~\ref{Act.ang}.

\noindent The most important part of the statement, for the proof of Theorem
\ref{Nek.1}, is ({\it iii}). 

\begin{remark}
In each of the domains $\cO_{i,j}^{(2)}$ the Hamiltonian, written in terms of
the action-angle variables, depends on the actions only. Furthermore
it is a function 
$$
h_{i,j}:\cA_{i,j}\to\R\ ,                                                
$$
which is analytic in the whole of $\cA_{i,j}$. This is a simple
consequence of the fact that the maps $\Phi_{i,j}^{(2)}$ are analytic
diffeomorphisms. 
\end{remark}

We come now to the three dimensional case.

Given $L\in\cI$ and arbitrary $E\in\R$, , we define
\begin{equation*}
\cL(L,E)\coloneqq\Bigl\{ (\bx,\bp)\, :\, L^2(\bx,\bp)=L^2\  \text{and}\  \hze (\bx,\bp)=E\Bigr\}\ .
\end{equation*}
The sets $\cL(L,E)$ can be empty or can have one or more connected
components. As in the $2-d$ case, we denote by 
$\cL_i(L,E)$ each compact connected component of
$\cL(L,E)$.
Define
\begin{equation*}
\cP_A\coloneqq\bigcup_{L\in\cI}\bigcup_{E\in\R}\bigcup_i\cL_i(L,E) \ .
\end{equation*}

We have the following result
\begin{lemma}
\label{az.ang.3}
Corresponding to any domain $\cO_{i,j}^{(2)}$ of Lemma \ref{az.ang}, there exists an open set $\cO_{i,j}\subset
\cP_A$, s.t., defining $\cS\coloneqq\cP_A\setminus \bigcup
_{i,j}\cO_{i,j}$ one has
\begin{itemize}
\item [(i)] 
$\cS$ is the union of a finite
  number of analytic hypersurfaces;
\item[(ii)] each of the domains $\cO_{i,j}$ can be covered by systems
  of generalized action-angle coordinates with the same action
  variables. Precisely, $\cO_{i,j}$ has the structure of a bifibration
  \begin{equation*}
\cO_{i,j}\mathop{\to}^{F_{i,j}}\cM_{i,j}\mathop{\to}^{\widetilde F_{i,j}}\cA_{i,j}\subset\R^2\ ,
  \end{equation*}
with the following properties
\begin{itemize}
\item[(1)] Every fiber of
  $\displaystyle{\cO_{i,j}\mathop{\to}^{F_{i,j}}\cM_{i,j}}$ is diffeomorphic to $\mathbb{T}^2${\rm\thinspace;}
\item[(2)] the bifibration is symplectic, i.e., every fiber of
  $\displaystyle{\cO_{i,j}\mathop{\to}^{F_{i,j}}\cM_{i,j}}$ has a
  neighborhood $U$ endowed with an analytic diffeomorphism 
\begin{equation}
\label{actvar}
\Phi_U:U\to p(U)\times q(U)\times\mathcal A_{i,j}\times \mathbb{T}^2\ni (p,q,I,\alpha)
\end{equation}
such that the level sets of $F_{i,j}$ coincide with the level sets
of $p\times q \times I$ and the symplectic form becomes
\begin{equation*}
dp\wedge dq+dI_1\wedge d\alpha_1+dI_2\wedge d\alpha_2\ .
\end{equation*}
The functions $(I_1,I_2)$ are globally defined on $\cO_{i,j}$.
\end{itemize}
\item[(iii)]  in each of the domains $\cO_{i,j}$, the Hamiltonian, as a
  function of the action variables $(I_1,I_2)$, coincides with the
  Hamiltonian $h_{i,j}$ of the planar central motion in $\cO_{i,j}^{(2)}$.
\end{itemize}
\end{lemma}

\section{Statement of the quasiconvexity result and of a further Nekhoroshev
  type result}\label{stat}

\begin{definition}
Let $\cA\subset\R^2$ be open. A function $h\colon\cA\to\R$ is said to be
quasiconvex at a point $I^*\in\cA$ if
\begin{equation*}
\left\{\biggl\langle{\eta}\,,\ \nparder{2}{h }{I}(I^*) \eta\biggr\rangle=0\quad \text{and} \quad 
\biggl\langle{\parder{h }{I}(I^*)}\,,\ {\eta}\biggr\rangle =0 
\right\}\ \Longrightarrow \eta=0\ .
\end{equation*}
\end{definition}

Our main technical result is the following theorem.

\begin{theorem}
\label{bert}
Assume that the potential $V(r)$ is real analytic and satisfies
(H1)--(H3), as in Theorem~\ref{Nek.1}, then one of the following two
alternatives hold:
\begin{itemize}
\item[(1)] for every $i,j$ there exists at most one analytic hypersurface
  $\cK_{i,j}\subset \cA_{i,j}$, s.t. $h_{i,j}$ is quasiconvex for all 
  $(I_1,I_2)\in\cA_{i,j}\setminus\cK_{i,j}$.
\item[(2)] there exists $k>0$ s.t. $V(r)=k r^2$ or $V(r)=-k/r$. 
\end{itemize}
\end{theorem}

\vskip10pt Theorem~\ref{Nek.1} immediately follows (see
Section~\ref{sec:4} for the details).
\begin{remark}
Both in the Harmonic and the Keplerian cases, the Hamiltonian only depends 
on one action, thus the steep Nekhoroshev theorem does not apply (see,
e.g.,~\cite{guzzo2016steep}).
\end{remark}

\bigskip
A remarkable feature of Nekhoroshev's theorem is that it
applies also to the case where the central force problem is in interaction
with a further slow system. 
Let $P\in C^\omega(\cP_A\times \cU^{(2n)})$, where $\cU^{(2n)}\subset
\R^{2n}$ is open, and consider the following Hamiltonian with $n+3$
degrees of freedom
\begin{equation}
  \label{fast}
H_\epsilon(\bx,\bp,\widehat{\bx},\widehat{\bp})\coloneqq
\frac{1}{\epsilon}\hze (\bx,\bp)+P (\bx,\bp,\widehat{\bx},\widehat{\bp})
\ .
\end{equation}
This describes, for example, a microscopic central system, so
characterized by very light particles subject to very intense forces,
interacting with a macroscopic subsystem. For example one could
consider $V$ to be the Lennard-Jones potential.

To give a precise statement define\footnote{In the forthcoming Theorem~\ref{nek.2} we will implicitly assume
$\cE(E_0,\bx,\bp)\subset\cU^{(2n)}$ for the values of $E_0$ that we
will consider.} 
\begin{equation*}
  \cE(E_0;{\bx},{\bp})\coloneqq\left\{(\widehat{\bx},\widehat{\bp})\in\R^{2n}\, :\, P({\bx}, {\bp},\widehat{\bx},\widehat{\bp})<E_0 \right\}\ .
\end{equation*}

\begin{theorem}
  \label{nek.2}
Assume that $V(r)$ is neither Harmonic nor Keplerian, let $\cC$ be
fixed as in Theorem~\ref{Nek.1} and let
$$ \cE(E_0)\coloneqq\bigcup_{({\bx},{\bp})\in\cC } \cE(E_0;{\bx},{\bp})\ .$$
Assume that there exists $E_0$ such that the function $P$ extends to a
bounded analytic function on some complex neighborhood of
$\cC\times\cE(E_0)$, then there exist positive $\epsilon^*$, $C_1$,
$C_2$, $C_3$ and $C_4$ with the following property: for
$|\epsilon|<\epsilon^*$, considering the dynamics of the Hamiltonian
system~\eqref{fast},  then, for any initial datum in $\cC\times\cE(E_0)$
one has that~\eqref{est.nek} hold for the times~\eqref{time}.
\end{theorem}

\begin{remark}
An interesting different application pertains the dynamics of the
 soliton of NLS in an external potential. In $\R^3$, consider
 \begin{equation}
   \label{nls}
{\it i}\dot\psi=-\Delta\psi-f'(\left|\psi\right|^2)\psi+\epsilon V\psi\ ,
 \end{equation}
 with $V$ a central potential of class Schwartz (the interesting case
 is that in which it is a potential well) and $f$ a smooth function
 with a zero of order one at the origin and with the further property
 that~\eqref{nls} admits a stable soliton solution for
 $\epsilon=0$. It is known (see, e.g.,~\cite{FJ}) that when
 $\epsilon\neq0$, in first approximation, the soliton moves as a
 mechanical particle described by the Hamiltonian $\hze $ with $V$
 substituted by a radial effective potential. Then in~\cite{BM16} it
 was shown that, for times longer then any power of $\epsilon^{-1}$,
 there is essentially no exchange of energy between the soliton and
 the rest of the field. Exploiting the result of the present paper one
 can add Nekhoroshev type techniques in order to show that, for the
 same time scale and for the majority of initial data, also the
 angular momentum of the soliton is almost conserved, and thus in
 particular it does not approach the bottom of the potential well. We
 avoid here a precise statement, since this would require a quite
 technical and long preparation.
\end{remark}

\section{Proof of Lemmas \ref{az.ang} and \ref{az.ang.3}}\label{Act.ang}

We start with the proof of Lemma~\ref{az.ang}. The first point is to
show that, except for at most finitely many values of $p_\theta \in\cI$,
the effective potential has only nondegenerate critical points and the
corresponding critical levels are distinct (see
Lemma~\ref{livelli}). 

\begin{remark}
\label{r.act.1}
Due to assumption (H2) there {\bf do not exist} constants $k_1,k_2$
s.t. the  potential has the form
\begin{equation*}
V(r)=\frac{k_1}{2r^2}+k_2\ .
\end{equation*}
\end{remark}

We also assume $L\not=0$.

\begin{lemma}\label{lemma4.1}
Let $(\bar r,\bar\ell)$ be such that $\bar r$ is an extremum of
$\Ve(.,\bar\ell)$. Then there exists an odd $n$, a neighborhood
$\cU$ of $0$ and a function
$r_0=r_0((\ell-\bar\ell)^{1/n})$ analytic in $\cU$,
s.t. $r_0((\ell-\bar\ell)^{1/n})$ is an extremum of
$\Ve(.;\ell)$. Furthermore $r_0(0)=\bar r$, and for any
$\ell\not=\bar\ell$ the extremum is nondegenerate.
\end{lemma}
\proof To fix ideas assume that $\bar r$ is a maximum. Of
course the theorem holds with $n=1$ if the maximum is nondegenerate. So,
assume it is degenerate. Then, since the function $\Ve(.,\bar \ell)$
is nontrivial there exists an odd $n>2$, s.t.  $\partial_r^{n+1}\Ve(\bar
r,\bar\ell)=a\not=0$. Thus we look for
$\delta=\delta(\xi)$ solving 
\begin{equation}
\label{impli}
F(\delta,\xi)\coloneqq\partial_r\Ve(\bar r+\delta,\bar
\ell+\xi^n)=\left[V'(\bar r+\delta)-\frac{\bar\ell}{(\bar
    r+\delta)^3}\right]-\frac{\xi^n}{(\bar
    r+\delta)^3 }=0\ .
\end{equation}
It is convenient to rewrite the square bracket as  
$$
\frac{a}{n!}\delta^n+R_0(\delta)\ ,
$$ where $R_0$ is an analytic function with a zero of order at least
$n+1$ at the origin. A short computation shows that we can rewrite~\eqref{impli} in the form
\begin{equation}
\label{impli.1}
\delta\left[\left(\frac{a}{n!}+\frac{R_0(\delta)}{\delta^n}\right)(\bar
  r+\delta)^3\right]^{1/n}=\xi\ ,
\end{equation}
which is in a form suitable for the application of the implicit
function theorem. Thus it admits a solution $\delta(\xi)$ which is
analytic and which has the form
\begin{equation*}
\delta(\xi)=\left(\frac{a}{n!}\bar r^3\right)^{-1/n}\xi+{\cal O}(\xi^2)\ .
\end{equation*}

It remains to show that for $\xi$ different from zero (small) the
critical point just constructed is nondegenerate. To this end we
compute the derivative with respect to $\delta$ of $F$
(c.f. equation~\eqref{impli}); we get
\begin{align}
\label{nondeg}
\partial_\delta F(\delta(\xi),\xi)&=\frac{a}{(n-1)!}[\delta(\xi)]^{n-1}+
R_0'(\delta(\xi))  +\frac{3\xi^n}{(\bar r+\delta)^4}
\\
\label{nondeg.1}
&=
\frac{a}{(n-1)!}\left(\frac{n!}{a\bar
  r^3}\right)^{\frac{n-1}{n}}\xi^{n-1} +{\cal O}(\xi^n)\ ,
\end{align}
which for small $\xi$ is non vanishing. \qed

\begin{lemma}
Let $\bar r$ be a degenerate critical point of $\Ve(.;\bar\ell)$ which
is neither a maximum nor a minimum. Then for $\ell$ in a neighborhood
of $\bar\ell$, the effective potential $\Ve(.;\ell)$ either has no
critical points in a neighborhood of $\bar r$, or it has a
nondegenerate maximum and a nondegenerate minimum which depend
smoothly on $\ell$.
\end{lemma}
\proof
 A procedure similar to that used to deduce the equation~\eqref{impli.1} leads to the equation 
\begin{equation*}
\delta^n\left(\frac{a}{n!}+\frac{R_0(\delta)}{\delta^n}\right)(\bar
  r+\delta)^3=\ell-\bar\ell\ ,
\end{equation*}
where {\it $n$ is now even} and the sign of $a$ is arbitrary. It is
thus clear that for $(\ell-\bar\ell)/a$ negative the critical point disappear.
When this quantity is positive then it is easy to see that two new
critical points bifurcate from $\bar r$. Using a computation similar
to that of equations~\eqref{nondeg} and~\eqref{nondeg.1} one sees that they are
a maximum and a minimum which are nondegenerate. \qed

\begin{corollary}
There exists a finite set $\cI_{s1}\subset\cI$ such that, $\forall p_{\theta}\in
\cI\setminus\cI_{s1}$ the effective potential $\Ve(.;p_{\theta}^2)$ has only
critical points which are nondegenerate extrema.
\end{corollary}
\proof Just remark that the values of $p_{\theta}$ for which  $\Ve(.;p_{\theta}^2)$ has
at least one degenerate critical point are isolated. Thus, due to the
compactness of $\overline{\cI}$ their number is finite. \qed

\begin{remark}
  \label{analytic}
The set $\cI\setminus\cI_{s1}$ is the union of finitely many open
intervals. The critical points of $\Ve$ are analytic functions of
$p_{\theta}^2$ in such intervals; furthermore they do not cross (at crossing
points their multiplicity would be greater then one, against
nondegeneracy). Therefore the number, the order and the nature of the
critical points is constant in each of the subintervals.
\end{remark}

The main structural result we need for the effective potential is the
following

\begin{lemma}
  \label{livelli}
There exists a finite set $\cI_{s}\subset\cI$ such that, $\forall p_{\theta}\in
\cI\setminus\cI_{s}$ the effective potential $\Ve(.;p_{\theta}^2)$ has only
critical points which are nondegenerate extrema and the critical
levels are all distinct. Furthermore each critical level does not
coincide with
\begin{equation*}
V^{\infty}\coloneqq\lim_{r\to\infty}V(r)\ .
\end{equation*}
\end{lemma}
\proof First we restrict to $\cI\setminus\cI_{s1}$, so that all the
critical points of $\Ve$ are nondegenerate. We concentrate on one of
the open subintervals of $\cI\setminus\cI_{s1}$
(cf. Remark~\ref{analytic}). Let $\ell\coloneqq p_{\theta}^2$, and let $r(\ell)$ be a
critical point of $\Ve(.;\ell)$. Consider the corresponding critical
level
$\Ve(r(\ell),\ell)$ and compute
\begin{equation}
  \label{cri.1}
\frac{d}{d\ell}\Ve(r(\ell),\ell)=\frac{\partial r}{\partial
  \ell}\frac{\partial \Ve}{\partial r}(r(\ell),\ell)+\frac{\partial
  \Ve}{\partial \ell}=\frac{1}{2r^2} \ ,
\end{equation}
where we used the fact that $r(\ell)$ is critical, so that
$\frac{\partial \Ve}{\partial r}(r(\ell),\ell)=0$ and the explicit
expression of $\Ve$ as a function of $\ell$. Thus the derivative~\eqref{cri.1} depends on $r$ only. It follows that if two critical
levels coincide, then their derivatives with respect to $\ell$ are
different, and therefore they become different when $\ell$ is
changed. It follows that also the set of the values of $\ell $ for
which some critical levels coincide is formed by isolated points, and
therefore it is composed by at most a finite number of points in each
subinterval. Of course a similar argument applies to the comparison
with  $V^{\infty}$.\qed

We are now ready for the construction of action-angle variables (and
the proof of Lemma~\ref{az.ang}).
Consider one of the connected subintervals of $\cI\setminus\cI_s$
and denote it by $\widetilde{\cI}$. We distinguish two cases: (1) the effective
potential has no local maxima for $p_{\theta}\in\widetilde{\cI}$; (2) the effective
potential has at least one local maximum for $p_{\theta}\in\widetilde{\cI}$.

We start by the case (1). The second action is $I_2\coloneqq p_{\theta}$, while the
first one is given by
\begin{equation}
  \label{az.1}
I_1=G(E,I_2)\coloneqq\frac{1}{\pi}\int_{r_{min}}^{r_{max}}\!\!\!\sqrt{2(E-\Ve(r;I_2^2))}dr\ ,\quad
E\in (\Ve(r_0;I_2^2),V^{\infty})\ ,
\end{equation}
where $r_{min}$ and $r_{max}$ are the solutions of the equation
$E=\Ve(r;I_2^2)$ and $r_0$ is the minimum of the potential.
Correspondingly the action $I_1$ varies in
$(0,G(V^{\infty},I_2))$. Thus the domain $\widetilde{\cO}^{(2)}$ is
\begin{equation*}
\widetilde{\cO}^{(2)}\coloneqq\left\{(r,p_r,\theta,p_{\theta})\, :\, p_\theta\in\widetilde{\cI}\ ,
(r,p_r)\in \bigcup_{E\in (\Ve(r_0;I_2^2),V^{\infty})}\cL^{(2)}(p_{\theta}, V^{\infty} )
\right\}\ ,
\end{equation*}
and the actions vary in 
\begin{equation*}
\widetilde{\cA}\coloneqq\left\{(I_1,I_2)\, :\, I_2\in\widetilde{\cI}\ ,\ I_1\in
(0,G(V^{\infty},I_2))\right\}\ .
\end{equation*}
In this domain the Hamiltonian is computed by computing $E$ as a
function of $I_1,I_2$ by inverting the function $G$ defined in~\eqref{az.1}.

In this case we define the affine map $\widetilde L_{p_\theta}$ of
Lemma~\ref{az.ang} by
\begin{equation}
  \label{affine}
\widetilde L_{p_\theta}\te\coloneqq \left[V^{\infty}-
  \Ve(r_0(p_\theta^2);p_\theta^2) \right] \te+ \Ve(r_0(p_\theta^2);p_\theta^2)\ .
\end{equation}

We consider now the case where the effective potential has at least
one local maximum. In this case, in general, there are several
different domains which are described by action-angle coordinates. To
fix ideas consider the case where $\Ve(.;p_{\theta}^2)$ has exactly
two minima $r_{01}>r_{03}$ and one maximum $r_{02}$ fulfilling
$\Ve(r_{02};I_2^2)<V^{\infty}$ (as in Fig.~\ref{fig1}). Then the level
sets $\cL(p_\theta,E)$, for $\Ve(r_{03}; I_2^2)<E<\Ve(r_{02};I_2^2)$ have two
connected components, in each of which one can construct the action
variables exactly by the formula~\eqref{az.1} (with a suitable
redefinition of $r_{min}$ and ${r_{max}}$). 

Then there is a further domain in which the action $I_1$ can be defined;
it corresponds to the level sets with $ \Ve(r_{02};I_2^2)<E<V^{\infty} $
above the local maximum. In this domain the action
is still given by the formula~\eqref{az.1}.
It is clear that in more general situations only one more kind of
domains of the phase space can exists: namely domains in which both
the minimal energy and the maximal energy correspond to the energies
of local maxima of the effective potential, but in this case the
construction goes exactly in the same way. 

Summarizing we have that the following
\begin{lemma}
Each of the domains $\cO_{i,j}^{(2)}$ in which a system of
action-angle variables is defined is the union for $I_2$ in an open
interval, say $\cI_{j}$ of level sets of $\heff $. The infimum over
$\cO^{(2)}_{i,j}$ of the energy $\heff $ is either a nondegenerate
maximum or a nondegenerate minimum of the effective potential. The
corresponding value of the radius will be denoted by $r_{0j}=
r_{0j}(I^2_2)$ and depends analytically on $I_2\in\cI_{j}$.
\end{lemma}

\noindent{\it Proof of Lemma~\ref{az.ang}.} So we have constructed
action-angle variables in a domain excluding the following sets, which we
define $\cS^{(2)}$ and which are analytic hypersurfaces:
\begin{itemize}
\item[(1)]
  $\displaystyle{\left\{(r,p_r,\theta,p_{\theta})\ :\ p_\theta=0\right\}}$ \thinspace;
\item[(2)]
  $\displaystyle{\left\{(r,p_r,\theta,p_{\theta})\ :\ p_\theta\in\cI_s\right\}}$  \thinspace;
\item[(3)]
  $\displaystyle{\left\{(r,p_r,\theta,p_{\theta})\ :\  
V_{0j}(p_{\theta})=\hze(r,p_r,p_{\theta})\right\}}$\thinspace.
\end{itemize}
where we denoted 
\begin{equation*}
V_{0j}(I_2)\coloneqq\Ve(r_{0j}(I_2^2),I_2^2)=V(r_{0j})+\frac{
    r_{0j}V'(r_{0j})}{2}\ ,
\end{equation*}
the corresponding critical level.
\qed

\smallskip
\noindent {\it Proof of Lemma~\ref{az.ang.3}}. First we fix a couple
of indexes $i,j$ as in the Lemma~\ref{az.ang}, and, for $L\in\cI_i$ and
$E\in L_{i,L}(0,1)$, define
\begin{equation*}
\cL_{i,j}(L,E)\coloneqq\bigl\{\, (\bx,\bp)\ :\ L^2(\bx,\bp)=L^2\  \text{and}\  \hze (\bx,\bp)=E\,\bigr\}\ .
\end{equation*}
Then simply define $\cO_{i,j}\coloneqq\bigcup_{L\in \cI_j}\cL_{i,j}(L,E)$. By
the Lemma \ref{az.ang} it is clear that the assumptions of Theorem~1
in~\cite{Fas05} are fulfilled and the results on the structure of the
phase space holds. To get
that the Hamiltonian has the same form as in the planar case, just
remark that the whole phase-space can be covered using two systems of
polar coordinates with $z$-axis ($\theta=0$) not coinciding. Using any
one of the two systems, one can introduce explicitly, by the classical
procedure, action-angle variables which turn out to be $I_2=L$ and
$I_1$ which is the action of the Hamiltonian system with $1$ degree of
freedom and Hamiltonian $H(r,p_r,L^2)$. Thus, $I_1$ has exactly the
same expression as in the planar case, but with $p_\theta^2$ replaced
by $L^2$.

It follows that the Hamiltonian which is computed by inversion of the
formula for $I_1$ has the same functional form as in the planar case.
\qed

\section{Proof of Theorem~\ref{bert}.}\label{sec:4}

From Lemma~\ref{az.ang.3}, it follows that it is enough to
consider the planar case; so we now study such a case.

\subsection{The quasiconvexity condition}\label{sec:qc-burger}

Before going to the heart of the proof we give a couple of equivalent
forms of the quasiconvexity condition. Consider a Hamiltonian
$h\colon\cA\to\R$, with $\cA\subset \R^2$, with two degrees of freedom, in action
variables and define
$$
\omega_1 = \parder{h}{I_1}\ ,\quad
\omega_2 = \parder{h}{I_2}\ .
$$

It is well-known that that considering systems with two degrees of freedom
quasiconvexity is equivalent to the non vanishing of the Arnold
determinant (see, e.g.,~\cite{RCD}), namely
\begin{equation*}
\def\frac{\dfrac}
\Dscr=\det
\begin{pmatrix}
  \nparder{2}{h }{I} & \left(\parder{h }{I}\right)^T\ \\
  \parder{h }{I}  & 0
\end{pmatrix}\ ,
\end{equation*}
that explicitly reads
\begin{equation*}
\Dscr=-\frac{\partial^2 h }{\partial I_1^2}\omega_2^2+2\frac{\partial^2
 h }{\partial I_1\partial I_2}\omega_1\omega_2-\frac{\partial^2
  h }{\partial I_2^2}\omega_1^2\ .
\end{equation*}
Thus, rearranging the terms appearing in $\Dscr$, it is straightforward
to see that, if $\omega_2$ does not vanish, the condition $\Dscr=0$
can be written as a Burgers equation,
\begin{equation}\label{eq:burger}
\parder{\nu}{I_1}  = \nu \parder{\nu}{I_2} \ ,\quad
\hbox{with }\nu=\frac{\omega_1}{\omega_2}\ ,
\end{equation}
which is a form convenient for the study of $\cD$ close to a minimum
of the effective potential. It is easy to see that, provided $L\not=0$,
close to a minimum of the effective potential $\omega_2$ does not
vanish.

\subsection{Domains bounded below by a minimum}\label{sec:min}

In this section we concentrate on domains $\cO_{i,j}^{(2)}$ such that the
infimum of the energy $\hze $ at fixed $I_2$ is a minimum of the
effective potential. In particular the minimum $r_{0j}$ of the
effective potential is nondegenerate. In this section, since the
domain is fixed, we omit the index $j$ from the various
quantities. Thus $\cA$ will be the domain of the actions, $h$ the
Hamiltonian written in action variables, $r_0$ the minimum of the
effective potential and $V_0$ the corresponding value.

The main result of this section is the following
\begin{lemma}
  \label{min}
Let $\cO_{i,j}^{(2)}$ be a domain where the infimum of the effective
Hamiltonian at fixed $I_2$ is a nondegenerate minimum of the effective
potential. Assume that the Arnold determinant vanishes in an open
subset of $\cO_{i,j}^{(2)}$, then the potential is either Keplerian or
Harmonic.
\end{lemma}

The rest of the section is devoted to the proof of such Lemma.

As explained in the introduction we exploit the remark that in
one-dimensional analytic systems the Birkhoff normal form converges in a
(complex) neighborhood of a nondegenerate minimum. This, together with
the essential uniqueness of the action variables in one-dimensional systems,
implies that, for any $I_2$, the Hamiltonian $h$, as a function of
$I_1$, extends to a complex analytic function in a neighborhood of
$I_1=0$ and that the expansion constructed through the one-dimensional
Birkhoff normal form is actually the expansion of $h(I_1,I_2)$ at
$I_1=0$. It follows that also $\cD$ extends to a complex analytic
function of $I_1$ in a neighborhood of 0.  Thus one has an expansion
\begin{align*}
h(I_1,I_2)  = h_0(I_2)  + h_1(I_2) I_1  + \ldots  +
h_n(I_2)I_1^n  +\ldots\ ,
\end{align*}
where the quantities $h_n$ can be in principle computed as functions
of the derivatives of $V$ at $r_0(I_2)$ and of $I_2$. 

Here we will proceed by an explicit construction using a symbolic
manipulator.  We recall that explicit computations of Birkhoff normal
forms have already been implemented numerically in different context
(see, e.g., \cite{Sansottera20131,2014-GioLocSan,2017-GioLocSan} in
the context of the $n$-body problem,
and~\cite{Sansottera2014,Sansottera2015} for the spin-orbit problem).
We refer to that works for the details about the implementation of the
Birkhoff normal form.  We only stress that here the main difference
consists in the fact that all the parameters are represented by
symbols and the coefficients are rational numbers: thus the Birkhoff
normal form is computed in exactly and without any numerical
approximation.

\begin{remark}
\label{r.min}
In
${\cO_{i,j}^{(2)}} $ there is a 1-1 correspondence between $I_2$ and
$r_0$, so each of the functions $h_n$ can be considered just as a
function of $r_0$ and of  the derivatives of $V$ at
$r_0$. Correspondingly the derivatives with
respect to $I_2$ can be converted into derivatives with
respect to $r_0$ through the rule
 \begin{equation}\label{eq:dIdr}
\parder{}{I_2} =
\frac{2}{(3+g(r_0)) \sqrt{r_0 V'(r_0)}}\,
\parder{}{r_0}\ ,
\end{equation}
where, following the notation introduced in \cite{RCD}, we have defined
\begin{equation}
\label{g}
g(r_0)\coloneqq\frac{{r_0} V''({r_0})}{V'({r_0})}\ .
\end{equation}
Furthermore, it is convenient to define
\begin{equation}
  \label{defR}
R(r_0,V'(r_0),g(r_0))\coloneqq\frac{2}{(3+g(r_0)) \sqrt{r_0 V'(r_0)}}\ .
\end{equation}
\end{remark}

\begin{remark}
The condition that the function $g$ is constant is equivalent to the fact that the potential is
homogeneous or logarithmic, precisely, one has
\begin{equation*}
  g(r_0)=c\quad \iff\quad
  \begin{cases}
 V(r) = \frac{k}{c+1} r^{c+1}\,,\ 
    k\in\reali\setminus \left\{0\right\}& \text{if} \ c\not=-1
    \\
 V(r) = k \ln r\,,\
    k\in\reali\setminus \left\{0\right\}& \text{if} \ c=-1
     \end{cases}
\end{equation*}
\end{remark}
Then one can use the expansion obtained through the Birkhoff normal form
to compute the expansion of $\nu\equiv\omega_1/\omega_2$ at the minimum, namely
\begin{align}
\label{expan.1}
\nu(I_1,I_2)  = \nu_{0}(I_2)  + \nu_{1}(I_2) I_1  + \ldots  + \nu_{n}(I_2)
I_1^n+\ldots\ .
\end{align}
\begin{remark}
\label{ordine}
Since the frequency is the derivative of the Hamiltonian with respect
to the action, in order to compute $\nu_j$ one has to know the
Birkhoff normal form at order $2(j+1)$ in the Cartesian variables,
which correspond to know the expansion of the Birkhoff normal form at
order $j+1$ in $I_1$. 
\end{remark}

The idea is now to impose that the Burgers equation~\eqref{eq:burger}
is satisfied up to the first order in $I_1$ identically as function
of $I_2$, namely to impose
\begin{align}\label{eq:nu1-2}
  \begin{aligned}
  \nu_{1} &= \nu_{0}\parder{\nu_{0}}{I_2}\ ,\\
  \nu_{2} &= \frac{1 }{2}\left(\nu_{1}\parder{\nu_{0}}{I_2} +\nu_{0}\parder{\nu_{1}}{I_2}\right)\ ,
\end{aligned}
\end{align}
and to consider such equations as the ones determining the
degenerate potentials. We will show that such equations admit the only
common solutions given by the Harmonic and the Keplerian potential. Of
course a degenerate Hamiltonian fulfills also higher order equations,
but they are not relevant for our result.

According to Remark~\ref{r.min}, we will consider all the
 $\nu_j$ as functions of $r_0$ instead of $I_2$ and convert
all the derivatives with respect to $I_2$ into derivatives with
respect to $r_0$ by using~\eqref{eq:dIdr}.

Finally, it is convenient to use, as much as possible, $g$ as an
independent variable (see~\eqref{g}) instead of $V$. We remark
that $V''({r_0})={g(r_0)V'({r_0})}/{r_0}$, which implies that
$\forall n\geq2$ the $n$-th derivative of the potential can be
expressed as a function of $r_0,
V'(r_0),g(r_0),g'(r_0),\ldots,g^{(n-2)}(r_0)$. We will systematically
do this.

There is a remarkable fact: writing explicitly the
equations~\eqref{eq:nu1-2}, it turns out that they are independent of
$V'$, thus they are only differential equations for $g$.

\bigskip
We report below the outline of the computations and the key
formul{\ae}.  The symbolic  manipulations have been implemented in
Mathematica\texttrademark and are available upon request to the
authors. 

\begin{remark}
\label{esatti}
Let us stress that in principle all the calculations could be done
without resorting to algebraic manipulation on a computer. We just
made use of symbolic manipulations because the expressions become soon
too cumbersome. We emphasize that no approximations are involved in
the manipulations, since all the coefficients are rational numbers.
\end{remark}

Firstly we computed explicitly $\nu_0$, $\nu_1$ and $\nu_2$, defined
by~\eqref{expan.1}, getting
$$
\begin{aligned}
  \nu_0&=\sqrt{3+g(r_0)}\ ,
  \\
\nu_1&=\nu_1(r_0, V'(r_0), g(r_0), g'(r_0),g''(r_0))\ ,\\
\nu_2&=\nu_2(r_0, V'(r_0), g(r_0), g'(r_0),g''(r_0),g^{(3)}(r_0),g^{(4)}(r_0))\ .\\
\end{aligned}
$$
The explicit forms of $\nu_1$ and $\nu_2$ are rather long and are
reported in the Appendix~\ref{app:conti}.

Secondly we can use the explicit forms of the functions $\nu_0$ and
$\nu_1$ so as to compute the r.h.s. of~\eqref{eq:nu1-2}, which will have
the form
\begin{align*}
R \nu_{0}\parder{\nu_{0}}{r_0} &=: G_1(r_0, V'(r_0), g(r_0), g'(r_0))\ ,
\\
\frac{1}{2}R \left(\nu_{1}\parder{\nu_{0}}{r_0}+\nu_{0}\parder{\nu_{1}}{r_0}\right) &=: G_2(r_0, V'(r_0), g(r_0), g'(r_0), g''(r_0), g^{(3)}(r_0))\ ,
\end{align*}
where $R$ is the expression defined in~\eqref{defR}.

Lastly, imposing $\nu_1=G_1$ and $\nu_2=G_2$, we get two differential
equations for $g$ that we are going to solve. Let us stress that the
first one is exactly the equation for $g$ appearing in Remark~3
of~\cite{RCD}.

The strategy, in order to find the common solutions, consists in taking
derivatives of the equation of lower order until one gets two
equations of the same order (forth order in $g$ in our case), then
one solves one of the equations for the higher order derivative and
substitutes it in the other one, thus getting an equation of order
smaller then the previous one; then one iterates. In our case the
final equation will be an algebraic equation for $g$, whose solutions
are just constants. The value of such constants correspond to the
Kepler and the Harmonic potentials, so the conclusion will hold.
\vskip10pt

In detail, we solve $\nu_1 = G_1$ for $g''(r_0)$ and $\nu_2=G_2$ for
$g^{(4)}(r_0)$, getting
\begin{equation}\label{eq:f2f4}
\begin{aligned}
g''(r_0)&=f_2(r_0, g(r_0), g'(r_0))\ ,\\
g^{(4)}(r_0)&=f_4(r_0, g(r_0), g'(r_0),g^{(3)}(r_0))\ ,
\end{aligned}
\end{equation}
where, in the second one, we also used $f_2(r_0, g(r_0), g'(r_0))$
to remove the dependence of $g^{(4)}$ on $g''(r_0)$. A similar
procedure will be done systematically. 

Starting from~\eqref{eq:f2f4}, we compute
$$
\ndiff{2}{f_2}{r_0} = F_4(r_0, g(r_0), g'(r_0), g^{(3)}(r_0))\ ,
$$
and solve the equation $F_4(r_0, g(r_0), g'(r_0), g^{(3)}(r_0)) = f_4(r_0, g(r_0), g'(r_0), g^{(3)}(r_0))$ for $g^{(3)}$, getting
$$
g^{(3)} = f_3(r_0, g(r_0), g'(r_0))\ .
$$
Starting again from~\eqref{eq:f2f4}, we compute
$$
\diff{f_2}{r_0} = F_3(r_0, g(r_0), g'(r_0))\ ,
$$
and solve the equation $F_3(r_0, g(r_0), g'(r_0)) = f_3(r_0, g(r_0), g'(r_0))$ for $g'$ getting
$$
g' = f_1(r_0, g(r_0))\ .
$$
Finally we compute
$$
\diff{f_1}{r_0} = F_2(r_0, g(r_0))\ ,
$$ and solve $F_2(r_0, g(r_0)) = f_2(r_0, g(r_0))$ for $g$. It is
remarkable that such an equation turns out to be independent of $r_0$,
so that the solutions for $g(r_0)$ are just isolated points, namely
constants. In~\cite{RCD} it was already shown that the only constants
solving $\nu_1=G_1$ are $-3,-2$ and $1$. The value $-3$ is excluded
according to Remark~\ref{r.act.1}, therefore the only remaining
potentials are the Keplerian and the Harmonic ones. This concludes the
proof of Lemma~\ref{min}.

\subsection{Domains bounded below by a maximum}\label{sec:max}

Consider now domains $\cO_{i,j}^{(2)}$ s.t. the infimum of the energy $H$ at a
fixed $I_2\in\widetilde{\cI}$ is a nondegenerate maximum of the effective
potential $\Ve$. Denote by $V_0=V_0(I_2)$ the value of the effective
potential at such a maximum delimiting from below the range of the
energy in $\cO_{i,j}^{(2)}$. 

\begin{lemma}
  \label{max}
Let $\cO_{i,j}^{(2)}$ be a domain s.t. the infimum of the effective Hamiltonian
at fixed $I_2$ is a nondegenerate maximum of the effective
potential, then the Arnold determinant vanishes in $\cO_{i,j}^{(2)}$ at most on an
analytic hypersurface.
\end{lemma}

The rest of the section is devoted to the proof of such a lemma. The
main tool for studying the limiting behavior of the action close to
the maximum $V_0$ is the following normal form theorem, which is a
slight reformulation of a simplified version of the main result
in~\cite{Gio01}.

\begin{theorem}\label{NForm}
  Let
$$W(r)=W_0 -\frac{\lambda^2}{2}r^2+\mathcal{O}(r^3)$$ be an analytic
  potential having a nondegenerate maximum at $r=0$; consider the 
  Hamiltonian
$$H(r,p_r)=\frac{p_r^2}{2}+W(r)$$ then, there exists an open
  neighborhood $\cV_0$ of $0$ and a near to identity canonical
  transformation $\Phi: \cV_0\ni(x,y)\mapsto
  (r,p_r)\in\cU_0\coloneqq\Phi(\cV_0)$ of the form
\begin{equation*}
\left\{
\begin{aligned}
r&=\frac{x}{\sqrt{\lambda}}+ f_1(x,y)\\
p_r&=\sqrt{\lambda}y+f_2(x,y)
\end{aligned}
\right.
\end{equation*}
with $f_1,f_2$ analytic functions which are at least quadratic in
$x,y$ and such that the Hamiltonian takes
the form
\begin{equation*}
h(x,y)=W_0+\lambda J + \sum_{i\geq 1} \lambda_i J^{i+1} \ ,\qquad
J\coloneqq\frac{y^2-x^2}{2}\ .
\end{equation*}
Furthermore, the series is convergent in $\mathcal{V}_0$.
\end{theorem}

The behavior of the action variable close to the maximum of the effective
potential is described by the following 

\begin{theorem}\label{formaI1}
Let $\bar E=E-V_0(I_2)$, then there exist two functions $\Lambda (\bar
E,I_2)$, $G_1\left(\bar E,I_2\right)$ such that the first action $I_1$
is given by
\begin{equation}\label{I1finale}
I_1=G(\bar E,I_2)\coloneqq-\Lambda(\bar E,I_2)\ln\bar E + G_1\left(\bar E,I_2\right)\ .
\end{equation}
Furthermore $\Lambda$ and $G_1$ are bounded in the
domain
\begin{equation}
  \label{dom.max}
\left\{(\bar E,I_2)\ :\ I_2\in\widetilde\cI\ ,\ \bar E\in [0, V_M(I_2)-V_0(I_2))
  \right\}\ ,
\end{equation}
where $V_M(I_2)$ is the maximal value of the
energy at fixed $I_2$ in $\cO_{i,j}^{(2)}$ and analytic in its interior.
Finally one has 
\begin{equation*}
\Lambda(\bar E,I_2) \coloneqq\frac{\bar E+\mathcal{F}(\bar E, I_2)}{\pi\lambda(I_2)}\ ,
\end{equation*}
with $\mathcal{F}$ having a zero of order $2$ in $(0,I_2)$ and
$\lambda^2=\lambda^2(I_2)\coloneqq-\frac{d^2 \Ve}{dr^2}(r_0)>0$.
\end{theorem}
\begin{remark}
The main point is that the lower extremum of the interval~\eqref{dom.max}
for $\bar E$ is included, thus~\eqref{I1finale} describes the actions
until the maximum. Furthermore, by~\eqref{I1finale} the
limit
$$I_{1 \, 0} \coloneqq \lim_{\bar E \rightarrow 0^+}G(\bar E,I_2)=G_1(0,I_2)$$
exists and is finite.
\end{remark}

\begin{remark}
Since $E\mapsto G( E-V_0(I_2),I_2)$ is a monotonically increasing
function for $E\in (V_0(I_2), V_M(I_2))$, there exists a function
$h(I_1,I_2)$ such that
$$G(h(I_1,I_2)-V_0(I_2),I_2)\equiv I_1\ .$$
Furthermore, by the implicit function theorem, $h$ is analytic in $I_1,I_2$ for $I_2\in\widetilde{\cI}$ and $I_1 > I_{1 \, 0}$.
\end{remark}

\noindent
{\it Proof of Theorem~\ref{formaI1}.}  Let $I_2\in\widetilde{\cI}$ and
consider the Hamiltonian~\eqref{H} with $p_{\theta} = I_2$. In the
whole construction $I_2$ will play the role of a parameter, thus we
will omit the dependence on $I_2$ and just consider the $(r,p_r)$
dependence.

Firstly, expanding at $r_0$, we obtain
$$H(r,p_r)=\frac{p_r^2}{2} +V_0- \frac{\lambda^2}{2}(r-r_0)^2 +
\mathcal{O}\left((r-r_0)^3 \right)\ .$$
Secondly, via the change of
variable $r'\coloneqq r-r_0$, (omitting the primes) we get
\begin{equation*}
H(r,p_r)=\frac{p_r^2}{2}+V_0 - \frac{\lambda^2}{2}r^2 + \mathcal{O}(r^3)\ .
\end{equation*}

Fix a value of the energy $E>V_0$, close enough to $V_0$ and denote by
$\gamma(E)$ the level curve corresponding to $H=E$.  Thus we have
\begin{equation*}
I_1=\frac{1}{2\pi} \int_{\gamma(E)} p_r dr=\frac{1}{\pi} \int_{\gamma^+(E)} p_r dr
\end{equation*}
where $\gamma^+(E)$ is the upper part of the level curve, namely, $\gamma|_{p_r>0}$.  

We split the domain of integration into two regions, precisely

\begin{equation}\label{I}
I_1=\frac{1}{\pi} \left[\int_{\gamma^+(E)\,\cap\, \mathcal{U}} p_r dr + \int_{\gamma^+(E)\,\cap\, \mathcal{U}^c} p_r dr\right]\ .
\end{equation}
where $\mathcal{U}$ is a neighborhood of the nondegenerate maximum
that will be fixed in a while. First, we remark that the second
integral does not see the critical point, so it is an analytic
function of $E$ until $V_0$. To analyze the first integral, we exploit
Theorem~\ref{NForm}.

Let us fix a small positive $x_1$ and let us consider the neighborhood
$\cV$ of 0 defined by
$$\R^2\supset\mathcal{V}\coloneqq(-x_1,x_1)\times \left(-\sqrt{\frac{4\bar
    E}{\lambda}+x_1^2},\sqrt{\frac{4\bar
    E}{\lambda}+x_1^2}\right)\ .$$ Provided $\bar E$ and $x_1$ are
small enough, one has $\mathcal{V} \subset \mathcal{V}_0$
(c.f. Theorem~\ref{NForm}). Let us define $\cU \coloneqq \Phi(\mathcal{V})$
and write $\gamma_+(E)\cap\cU$ in the variables $(x,y)$, parametrizing
it with $x\in (-x_1,x_1)$. To this end remark that, since the
Hamiltonian $H$ is a function of $J$ only, namely
$$H= V_0+\lambda J+ \mathcal{G}(J)$$
where $\mathcal{G}(J)=\sum_{i\geq 1}\lambda_iJ^{i+1}$, by the implicit function theorem, there exists an analytic function $\cF(\bar E)$ having a zero of order $2$ at $0$ and such that
$$J= \frac{\bar E + \cF(\bar E)}{\lambda}$$ and, therefore,
$\gamma_+(E)$ can be written in the form $(x,y(x))$ with
\begin{equation*}
y(x)\coloneqq\sqrt{\frac{2(\bar E + \cF(\bar E))}{\lambda}+ x^2}\ .
\end{equation*}
To compute the first integral in~\eqref{I}, we remark that since
$\Phi$ is canonical and analytic in a neighborhood of the origin,
there exists a function $S(x,y)$, analytic in a neighborhood of the
origin, such that
$$p_r d r = y dx + d S\ ,$$
thus, we have
$$\int_{\gamma_+(E)\cap \cU } p_r dr = \int_{-x_1}^{x_1} y dx + S(x_1,
y(x_1)) -S(-x_1,y(-x_1))\ .$$ Since $x_1$ is fixed, the terms involving
$S$ are analytic functions of $\bar E$. Thus, we only compute the
first integral, namely,
$$
\int_{-x_1}^{x_1}y d x = \int_{-x_1}^{x_1} \sqrt{\frac{2(\bar E + \cF(\bar E,I_2))}{\lambda}+ x^2} dx
$$
which takes the form
\begin{equation*}
\begin{aligned}
  \int_{-x_1}^{x_1}y d x&=\frac{2(\bar E + \cF(\bar E))}{\lambda} \ln\left(x_1 + \sqrt{x_1^2 +\frac{2(\bar E + \cF(\bar E))}{\lambda} }\right)\\
  &\qquad+ x_1 \sqrt{x_1^2 + \frac{2(\bar E + \cF(\bar E))}{\lambda} }\\
&\qquad-\frac{\bar E + \cF(\bar E)}{\lambda}  \ln \left(\frac{2(\bar E + \cF(\bar E))}{\lambda} \right) .
\end{aligned}
\end{equation*}
It is easy to see that the first two terms are analytic in a
neighborhood of $0$.

Rewriting the third term as
$$-\frac{\bar E + \cF(\bar E)}{\lambda} \ln \bar E- \ln
\left(\frac{2}{\lambda}+ \frac{2\cF(\bar E)}{\bar E}\right)$$ and
remarking that the second function is analytic in a neighborhood of
$0$, we get the result. Formula~\eqref{I1finale} is obtained by
reinserting the dependence on $I_2$.  \qed

\vskip10pt

We come now to the Arnold determinant. Firstly, we will work in the region
$\bar E >0$ so that the function $G$ is regular and the implicit
function theorem applies and allows to compute $h$ and its
derivatives. Then, we will study the limit $\bar E \rightarrow 0^+$.

By the implicit function theorem, the frequency $\omega_1$ is given by
\begin{equation}\label{W1}
\omega_1 = \parder{h}{I_1}=\left ( \parder {G}{\bar
  E}\right)^{-1}=:\cW_1(\bar E,I_2)\ .
\end{equation}

\begin{remark}\label{formaF}
Let $f$ be an analytic function of the form 
$$f=f(\bar E,I_2)=f(h(I_1,I_2)-V_0(I_2),I_2)\ ,$$
then, 
$$\diff{f}{I_2}\coloneqq\parder{f}{\bar E}\parder{\bar E}{I_2}+\parder{f}{I_2}=\parder{f}{\bar E}\left(\omega_2 -\parder{V_0}{I_2}\right)+\parder{f}{I_2}.$$
\end{remark}

\noindent It follows from Remark~\ref{formaF} and the implicit function theorem that the frequency $\omega_2$ is given by
$$\omega_2=\parder{h}{I_2}=-\parder{G}{I_2}\cW_1+\parder{V_0}{I_2}.$$
Thus, it is worth introducing a function $\cW_2$ defined by
\begin{equation}\label{W2}
\cW_2(\bar E,I_2)\coloneqq-\parder{G}{I_2}\cW_1+\parder{V_0}{I_2}.
\end{equation}

\begin{proposition} 
Let $h\colon \mathcal{A} \rightarrow\mathbb{R} $ be the Hamiltonian in two
degrees of freedom written in action-angle coordinates, then the
Arnold determinant can be rewritten in terms of $G$ and $\cW_1$ as

\begin{equation}\label{newArn}
{\cD}=  -\cW_1\parder{\cW_1}{\bar E}\left(\parder{V_0}{I_2}\right)^2
+
2\cW_1\parder{\cW_1}{I_2}\parder{V_0}{I_2}+\cW_1^3\nparder{2}{G}{I_2}-\cW_1^2
\nparder{2}{V_0}{I_2}
\end{equation}
\end{proposition}

\begin{proof}
By exploiting the formul{\ae}~\eqref{W1},~\eqref{W2} and Remark~\ref{formaF}, we compute the second derivatives of the Hamiltonian
$h$. We have

\begin{align*}
\nparder{2}{h}{I_1} &= \parder{\cW_1(\bar
  E,I_2)}{I_1}=\parder{\cW_1}{\bar
  E}\parder{h}{I_1}=\cW_1\parder{\cW_1}{\bar E} 
\\
\frac{\partial^2 h}{\partial I_1 \partial I_2} &= \diff{\cW_1(\bar
  E,I_2)}{I_2}
\\
\nparder{2}{h}{I_2} &=\diff{\cW_2(\bar E,I_2)}{I_2}= \diff{}{I_2}\left( -\cW_1\parder{G}{I_2}+\parder{V_0}{I_2}\right)\\
&=-\diff{\cW_1}{I_2}\parder{G}{I_2}-\cW_1\diff{}{I_2}\left(\parder{G}{I_2}\right)
+ \nparder{2}{V_0}{I_2}
\end{align*}

We can write the three terms of
the Arnold determinant separately as
\begin{equation}\label{D1}
{\cD_1}= -\cW_1 \cW_2^2\parder{\cW_1}{\bar E},
\end{equation}
\begin{equation}\label{D2}
{\cD_2}=2\cW_1\cW_2\diff{\cW_1}{I_2},
\end{equation}
\begin{equation}\label{D3}
{\cD_3}=\cW_1^2\parder{G}{I_2}\diff{\cW_1}{I_2}+\cW_1^3\diff{}{I_2}\left(\parder{G}{I_2}\right)
-\cW_1^2 \nparder{2}{V_0}{I_2}.
\end{equation}
And, gathering together the expressions~\eqref{D1},~\eqref{D2} and~\eqref{D3}, after simple computation, we obtain 
$${\cD} = -\cW_1\parder{\cW_1}{\bar E}\left(\parder{V_0}{I_2}\right)^2
+
2\cW_1\parder{\cW_1}{I_2}\parder{V_0}{I_2}+\cW_1^3\nparder{2}{G}{I_2}-\cW_1^2
\nparder{2}{V_0}{I_2}$$ This concludes the proof.
\end{proof}

\begin{proposition}
The Arnold determinant diverges as $\bar E$ tends to zero.
\end{proposition}

\begin{proof}
Due to the structure~\eqref{I1finale} of $G$, it is easy to see that
$\nparder{2}{G}{I_2}$ is bounded as $\bar E \rightarrow 0^+$ (remark
that $\parder{G}{I_2}$ means derivative with respect to the second
argument).  Therefore, since $\nparder{2}{V_0}{I_2}$ is a regular
function of $I_2$ and $\cW_1 \rightarrow 0$ as $\bar E$ approaches
zero, we have that
$$\lim_{\bar E \rightarrow 0^+}\left( \cW_1^3\nparder{2}{G}{I_2}-\cW_1^2 \nparder{2}{V_0}{I_2}   \right) =0.$$
Let us now concentrate on the analysis of the remaining terms of~\eqref{newArn}. The asymptotic behavior of the function $\cW_1$ is given by
$$\cW_1 \sim - \frac{\pi \lambda}{\ln \bar E}.$$
Concerning the derivatives, we have
$$\parder{\cW_1}{I_2} = -\left(\parder{G}{\bar
  E}\right)^{-2}\frac{\partial^2 G}{\partial I_2 \partial \bar
  E}\mathop{\to}^{\bar E \rightarrow 0^+  }  0\quad
\Longrightarrow\quad 
\lim_{\bar E \rightarrow 0^+}\left( 
2\cW_1\parder{\cW_1}{I_2}\parder{V_0}{I_2}\right)=0\ .
$$ Concerning the first term, using
$$\parder{\cW_1}{\bar E} \sim \frac{\pi \lambda}{\bar E\ln^2 \bar E} \ ,$$
we have that it behaves as
$$\frac{\pi^2\lambda^2}{\bar E \ln^3 \bar E}
\left(\parder{V_0}{I_2}\right)^2$$ which diverges to infinity as $\bar
E \rightarrow 0^+$. This concludes the proof.
\end{proof}

Lemma~\ref{max} follows from the fact that $\cD$ is a
nontrivial analytic function in $\cA_{i,j}$.

\appendix

\section{A new proof of Bertrand's Theorem}\label{Bert}

We give here a new proof of Bertrand's Theorem. Actually we give here
a local version of the Theorem (in the spirit of \cite{FK04}). The
obtained result is slightly stronger then the usual one.

\begin{theorem}[Bertrand]
Consider the planar central force problem with an analytic potential
$V$ s.t. a stable circular orbit exists. Denote by $\gamma$ the phase
space trajectory of such a circular orbit. If there exists a
neighborhood of $\gamma$ in which all the orbits are
periodic then the potential is either Keplerian or Harmonic.
\end{theorem}

\begin{proof}
First remark that a stable circular orbit is a minimum of the
effective potential corresponding to a particular value of the angular
momentum. By modifying the value of the angular momentum (according to
Lemma 4.1) one reduces to the case where the minimum is
nondegenerate. Consider now the quantity $\parder{\nu}{I_1}-\nu
\parder{\nu}{I_2}$: in Section~\ref{sec:qc-burger} we proved that it
vanishes identically only in the Keplerian and in the Harmonic
cases. Assume by contradiction that this is not the case. Then
$\parder{\nu}{I_1}-\nu \parder{\nu}{I_2}$ is a nontrivial function of
the actions, and thus the ratio $\nu(I_1,I_2)\equiv
\frac{\omega_1}{\omega_2}$ is also a non-constant function. It follows
that there exist $I_1,I_2$ such that $\nu$ is irrational and thus on
the corresponding torus the motion is not periodic, against the
assumption.
\end{proof}

\section{A couple of formul{\ae}}\label{app:conti}
We report here for completeness the explicit formul{\ae} of the
functions introduced in the proof of Lemma~\ref{min}.
$$
\footnotesize
\begin{aligned}
  \nu_1=&\Bigl(36+9 {r_0}^2 g''({r_0})-g({r_0})^2 \left({r_0} g'({r_0})+26\right)
  +g({r_0}) \left(3 {r_0}^2 g''({r_0})+7 {r_0} g'({r_0})+6\right)\\
&\ +30 {r_0}
g'({r_0})-5 {r_0}^2 g'({r_0})^2 -2 g({r_0})^4-14 g({r_0})^3\Bigr)
\Big/
\left( 24 {r_0}^{3/2} (g({r_0})+3)^2 \sqrt{V'({r_0})} \right) \ ,
\end{aligned}
$$

$$
\footnotesize
\begin{aligned}
  \nu_2=&\Bigl(-235 {r_0}^4 g'({r_0})^4+2604 {r_0}^3 g'({r_0})^3+4 g({r_0})^6 \left(17 {r_0}
  g'({r_0})+759\right)\\
  &\ +g({r_0})^5 \left(-84 {r_0}^2 g''({r_0})+600 {r_0} g'({r_0})+13408\right)+54 {r_0}^2
  g'({r_0})^2 \left(25 {r_0}^2 g''({r_0})-142\right)\\
  &\ -27 \left(17 {r_0}^4 g''({r_0})^2-456 {r_0}^2 g''({r_0})-24
  \left({r_0}^4 g^{(4)}({r_0})+12 {r_0}^3 g^{(3)}({r_0})-62\right)\right)\\
  &\ +g({r_0})^4 \left(-48 {r_0}^3
  g^{(3)}({r_0})-1344 {r_0}^2 g''({r_0})+129 {r_0}^2 g'({r_0})^2+624 {r_0} g'({r_0})+33500\right)\\
  &\ -108 {r_0}
  g'({r_0}) \left(14 {r_0}^3 g^{(3)}({r_0})+101 {r_0}^2 g''({r_0})+260\right)\\
  &\ +2 g({r_0})^3 \Bigl(699 {r_0}^2
  g'({r_0})^2+{r_0} g'({r_0}) \left(81 {r_0}^2 g''({r_0})-4732\right)\\
  &\ +12 \left({r_0}^4 g^{(4)}({r_0})-6
  {r_0}^3 g^{(3)}({r_0})-296 {r_0}^2 g''({r_0})+1705\right)\Bigr)\\
  &\ -g({r_0})^2 \Bigl(94 {r_0}^3 g'({r_0})^3-4053
  {r_0}^2 g'({r_0})^2+3 \Bigl(17 {r_0}^4 g''({r_0})^2+4680 {r_0}^2 g''({r_0})\\
  &\ -36 \left(2 {r_0}^4
  g^{(4)}({r_0})+12 {r_0}^3 g^{(3)}({r_0})+11\right)\Bigr)\\
  &\ +12 {r_0} g'({r_0}) \left(14 {r_0}^3 g^{(3)}({r_0})+20
  {r_0}^2 g''({r_0})+3291\right)\Bigr)\\
  &\ +2 g({r_0}) \Bigl(293 {r_0}^3 g'({r_0})^3+9 {r_0}^2 g'({r_0})^2 \left(25
   {r_0}^2 g''({r_0})+28\right)\\
   &\ -9 \left(-36 {r_0}^4 g^{(4)}({r_0})-360 {r_0}^3 g^{(3)}({r_0})+17 {r_0}^4
   g''({r_0})^2+198 {r_0}^2 g''({r_0})+2904\right)\\
   &\ -9 {r_0} g'({r_0}) \left(56 {r_0}^3 g^{(3)}({r_0})+323
   {r_0}^2 g''({r_0})+3216\right)\Bigr)\\
   &\ +20 g({r_0})^8+376 g({r_0})^7
   \Bigr)
\Big/
\Bigl(2304 {r_0}^3 (g({r_0})+3)^{9/2} V'({r_0})\Bigr) \ .
\end{aligned}
$$

\end{document}